\documentclass[12pt,
showpacs,
floatfix,amsfonts]{revtex4}

\usepackage{amsmath,amsthm,graphics,graphicx,eucal,amscd,amsfonts,amssymb}
\usepackage{epsfig}
\usepackage{graphicx}
\usepackage{amsmath}
\usepackage{amssymb}
\usepackage{dcolumn}

\newtheorem{prop}{Proposition}

\newcommand{\tA}{\tilde{A}}

\newcommand{\p}{\partial}

\begin{document}

\title{Field patterns in periodically modulated optical parametric amplifiers and oscillators}

\author{V. A. Brazhnyi$^{1}$}
\email{brazhnyi@cii.fc.ul.pt}

\author{V. V. Konotop$^{1,2,3}$}
\email{konotop@cii.fc.ul.pt}

\author{S. Coulibaly$^4$}
\email{Saliya.coulibaly@phlam.univ-lille1.fr}

\author{M. Taki$^4$}
\email{Abdelmajid.Taki@phlam.univ-lille1.fr}

\affiliation{
$^1$Centro de F\'{\i}sica Te\'orica e Computacional,
Universidade de Lisboa, Complexo Interdisciplinar, Avenida Professor Gama
Pinto 2, Lisboa 1649-003, Portugal\\
$^2$Departamento de F\'{\i}sica,
Universidade de Lisboa, Campo Grande, Ed. C8, Piso 6, Lisboa
1749-016, Portugal
\\
$^3$Departamento de Matem\'aticas, E. T. S. Ingenieros Industriales,
Universidad de Castilla-La Mancha, 13071 Ciudad Real,
Spain
\\
$^4$ Laboratoire de Physique des Lasers, Atomes et Mol\'ecules
(PHLAM), Centre d'Etudes et de Recherches Lasers et Applications
(CERLA),UMR-CNRS 8523 IRCICA, Universit\'e des Sciences et
Technologies de Lille, 59655 Villeneuve d'Ascq Cedex, France }


\begin{abstract}
Spatially localized and periodic field patterns in  periodically
modulated optical parametric amplifiers and oscillators are studied.
In the degenerate case (equal signal and idler beams) we elaborate
the systematic method of construction of the stationary localized
modes in the amplifiers, study their properties and stability. We
describe a method of constructing periodic solutions in optical
parametric oscillators, by adjusting the form of the external driven
field to the given form of either signal or pump beams.
\end{abstract}

\maketitle

\section{Introduction}

{\bf
Parametric amplification together with stimulated emission are
fundamental mechanisms for generation of coherent radiation. Laser
action is based on the latter whereas the former is the basis of
optical parametric oscillations (OPOs). Even though, laser systems are
encountered worldwide thanks to their intense commercialization,
OPOs are, nowadays, more and more
spreading (and not only in laboratories) owing, essentially, to
their tunability on a wide range of frequencies. In fact, for a
laser the emitted frequency is fixed once and for all to the
manufacturing. In consequence, a pile of several lasers of different
frequencies is needed when the operating system involves a whole
range of frequencies. We thus measure the considerable advancement
that OPOs bring to modern optical technology where large frequency
variations are desired. Although the basic principles of optical
parametric amplifiers (OPAs) have been known for over 40 years,
OPOs developed quickly only in the last decade
mostly for technological reasons \cite{Piskarskas}.
The use of periodically varying media greatly enriches the diversity of observable phenomena either in OPAs or in OPOs, allowing for existence of coherent field patterns, the study of which is the main goal of the present paper.
}

Compared to lasers, OPOs have received much less attention in spite of their
strong interest both on the fundamental and on the technological
sides \cite{Byer}. We recall that these are very frequency agile coherent sources
with a wide range of possible applications including range finding,
pollution monitoring and tunable frequency generation. They are also
the key element for the production of twin photons and the
realization of fundamental quantum optics experiments
\cite{Kolobov}.

Basically an optical parametric amplifier (OPA) generates light via a
three-wave mixing process in which a nonlinear crystal subjected to
a strong radiation at frequency $\omega _{p}$ (pump beam) radiates
two coherent fields at frequencies $\omega _{s}$ (signal beam) and
$\omega _{i}$ (idler beam) such that the energy conservation law
$\omega _{p}=\omega _{s}+\omega _{i}$ is satisfied. This energy
conservation criterion may be interpreted in terms of photons where
one photon at frequency $\omega _{p}$ is converted into two photons
at frequencies $\omega_{s}$ and $\omega_{i}$. This process is most
efficient when the phase matching condition is fulfilled. It states
that optical parametric amplification is favored when the three
interacting waves keep constant relative phases along their
propagation inside the crystal. This implies ${ \vec{k}}_p={
\vec{k}}_s+{ \vec{k}}_i,$ or equivalently $n_{p}(\omega _{p})\omega
_{p}=n_{s}\left( \omega _{s}\right) \omega _{s}+n_{i}\left( \omega
_{i}\right) \omega _{i}$ where $\vec{k}_{j}$ and $n_{j}$ are the
wavevector and the refractive index at frequency $\omega _{j}$
respectively ($j=p,s,i$).

OPOs have recently appeared as physical systems whose modeling
generates specific problems sharing common grounds with very general
open questions such as those related to the appearance of complexity
in spatially extended systems \cite{Hohenberg}. In fact, OPOs have
become one of the most active fields in nonlinear optics not only
for the richness in nonlinear dynamical behaviors \cite{REFSOPO} but
also for the potential applications of OPO devices
\cite{Piskarskas}, including low-noise measurements and defection
\cite{Kimble,Longhi}. When driven by an intense external field, both
OPAs; i.e. corresponding to a single pass through the quadratic
crystal, and OPOs show a number of remarkable features among them
are the localized structures (LS) or spatial solitons
\cite{Firth,Conti}. These nonlinear solutions are generated in the
hysteresis loop involving stable homogeneous states and the
nonhomogeneous periodical branches of solutions. The latter that
initiate  the LS formation result from the modulational instability
(often called Turing instability) and have been intensively  studied
in nonlinear optical cavities \cite{JOB2004}. Stability of LSs and
effects of spatial inhomogeneities on the LS dynamics have been
addressed in Refs. \cite{Skryabin12,Fedorov}.

Very recently another type of LS in the form of
dissipative ring-shaped solitary waves have been generated in the
regime where the steady state solutions are stable with respect to
the modulational instability. Indeed, it has been shown in
\cite{Taki-Springer} that OPOs can continuously generate spatially
periodic dissipative solitons with an intrinsic wavelength. These
modulations spontaneously develop from \emph{localized }
perturbations of the \emph{unstable} homogeneous steady state that
separates the two stable states of an hysteresis cycle. This
constitutes the counterpart of Turing spontaneous modulations
initiated by \emph{extended} perturbations. They occur in the wings
of the 2D traveling flat top solitons (fronts, domain walls or
kink-antikink) and give eventually rise to ring-shaped propagating
dissipative solitons. Such non-Turing periodic dissipative solitons
have also been predicted in liquid crystal light valve nonlinear
optical cavity \cite{Durniak}.

In the more general context of optics the soliton ability to
self-confine light beam power makes it a promising candidate for
future applications in information technology \cite{Nature06} and
image processing \cite{Nature02,Science97}. In this context, an
issue of important significance is the development of a research
activity aiming at taking advantage of more sophisticated geometries
of cavities~\cite{yulin} and periodic modulations of the refractive
index~\cite{Conti,All} to generate, control, and stabilize optical
LS. The existence of quadratic gap solitons in periodic structures,
having spatial extension much bigger than the period of the
structure, and thus well described within the framework of the
multiple scale expansion, was studied in \cite{Conti}. The control
and stabilization of LS are possible since periodic transverse
modulations are known, with the advent of photonic crystals, to
profoundly affect spatial soliton properties \cite{All}.

In this paper we address the problem of the formation and the
dynamics of localized structures in presence of transverse periodic
modulations of the refractive index in both OPAs and OPOs,
concentrating on field patterns whose spatial scales are of order of
the period of the medium. This is actually the main distinction of
our statement from the previous studies exploring either slow
envelope LSs~\cite{Conti} or the tight-binding limit leading to
nonlinear lattice equations~\cite{Egorov}. Indeed, optical
parametric wave-mixing systems are one of optical devices that can
operate either as a conservative (OPAs) or a dissipative (OPOs)
system. We have analytically investigated the existence and
stability of conservative LS (or solitons) that appear in OPAs and
the dissipative periodic patterns occurring in a parametric
oscillation regime (OPOs). We have first identified the simultaneous
conditions required on index modulation together with the incident
pump for stationary LS to exist by means of Bloch waves approach. It
results from our analysis that in a degenerate configuration (the
signal and idler are identical) signal and pump fields can be found
in form of phase locked LS (real solutions). They are strongly
stable and located at the minimum of the index periodic modulation.
To extend the analytical analysis to higher nonlinear regimes we
have performed numerical simulations that allowed characterizing the
different types of stable LS that can be supported by the system
under transverse periodic modulations.

The paper is organized as follows. In Sec. II we briefly recall the
two-wave model governing the spatio-temporal evolution of the slowly
varying envelopes of the pump and the signal waves including
diffraction and transverse periodic modulation of the refractive
index. The general conditions required for the existence and
stability of conservative LS appearing in OPAs, by means of Bloch
waves approach, are given. Extension of Bloch waves approach to take
into account dissipative effects stemming from cavity losses is
carried out in Section III, for a degenerate OPO with transverse
periodic modulations, where for construction of periodic solutions we use a kind of inverse engineering allowing us to determine the external field from a given structure of the nonlinear Bloch state. Concluding remarks are summarized in the last
section.

\section{Localized modes in periodically modulated
optical parametric amplifiers}

\subsection{The model}

We start with the model describing optical parametric amplification
 where we take into account diffraction and the transverse
periodic change of the refraction index
($n=n_j^{(0)}+n_j^{(1)}(\vec{r})$, $j=0,1$). In the degenerate
configuration, the signal and idler fields are identical leading to
the resonant frequency conversion $\omega_0=2\omega_1$ where
$\omega_0$ and $\omega_1$ are the pump and signal frequencies
respectively. Considering beam propagation along the $z$-direction,
the nonlinear interaction of the two waves propagating in the
crystal is governed by the system~\cite{Tlidi}:
\begin{subequations}
\label{eq_12}
\begin{eqnarray}
\label{eq1}
    &&\partial_z \alpha_0=\frac i2\nabla^2\alpha_0+i\tilde{\varepsilon}_0(\vec{r})\alpha_0+i\alpha_1^2e^{i\Delta kz}\,,
    \\\label{eq2}
    &&\partial_z \alpha_1=i\nabla^2\alpha_1+ i\tilde{\varepsilon}_1(\vec{r})\alpha_1+2i\alpha_0 \bar{\alpha}_1 e^{-i\Delta kz}\,,
\end{eqnarray}
\end{subequations}
where $\alpha_0$, $\alpha_1$ are the envelopes of the pump and
the idler, respectively. $\vec{r}=(x,y)$, $\nabla$ stands for the
transverse gradient, and
$\tilde{\varepsilon}_j=\omega_j^2[n_j^{(1)}]^2/c^2k_j$.
Hereafter an overbar stands for the complex conjugation.

We will be interested in periodically varying refractive indexes.
Let us further specify the model considering $n_j^{(1)}$ to be independent on frequency, i.e. assuming
\begin{eqnarray}
\label{eps01}
\tilde{\varepsilon}_0(\vec{r})=2 \tilde{\varepsilon}_1(\vec{r})\equiv \tilde{\varepsilon}(\vec{r}).
\end{eqnarray}
This assumption is not essential for construction of the field
patterns, reported below, and is introduced only for the sake of
concreteness, as the approaches we will consider involve numerical
simulations. Thus, the theory developed here is straightforwardly
generalized to any relation between $\tilde{\varepsilon}_0(\vec{r})$
and $\tilde{\varepsilon}_1(\vec{r})$, provided that their spatial
periods coincide (as it happens in a typical experimental
situation).

For perfect phase-matching, $\Delta k=0$, the system (\ref{eq_12}),
(\ref{eps01}) can be written in the Hamiltonian form
$\partial_z\alpha_j=i\delta H/\delta \bar{\alpha}_j$ with the
Hamiltonian
\begin{eqnarray}
    H&=& \int  \left[-\frac 12 |\nabla \alpha_0|^2- |\nabla \alpha_1|^2 \right. \nonumber\\
    &+&\left. \tilde{\varepsilon}\left(|\alpha_0|^2+\frac 12 |\alpha_1|^2\right)+\bar{\alpha}_0 \alpha_1^2+\alpha_0\bar{\alpha}_1^2\right]d^2\vec{r} \label{H}
\end{eqnarray}
being an integral of "motion": $dH/dz=0$. Another integral is given
by the total power $P=2P_0+P_1$, where we have introduced a notation
$P_j=\int |a_j|^2d^2\vec{r}$ for the power of the $j$-th component,
i.e. $dP/dz=0$. In what follows, existence domain and properties of
non-uniform stationary solutions will be investigated.

\subsection{General properties of stationary solutions}
\label{sec:gen_prop}

For the sake of simplicity, we address in the present work the
situation of $y$-independent spatial patterns, i.e. we assume
$n_j^{(1)}\equiv n_j^{(1)}(x)$.  Without loss of generality we can
impose that $n_j^{(1)}(x)$ is a $\pi$-periodic function (this can be
always achieved by proper renormalization of the coordinate units).
Then it is convenient to  introduce a periodic function
$V(x)\equiv -2\tilde{\varepsilon}(x)$,  $V(x+\pi)=V(x)$, and to define two
linear operators
\begin{eqnarray}
\label{lin_op}
    {\cal L}_j\equiv -\frac{d^2}{dx^2}+\frac{1}{4^{j}}V(x), \quad j=1,2
\end{eqnarray}
and the respective eigenvalue problems
\begin{eqnarray}
 {\cal L}_j\varphi_{\nu q}^{(j)}(x) ={\mathcal E}_{\nu q}^{(j)}\varphi_{\nu q}^{(j)}(x).
\label{eq:1deigen}
\end{eqnarray}
Here $\varphi_{\nu q}^{(j)}(x)$ are the linear Bloch waves, with the index $\nu=0,1,2,...$ standing for the number of the allowed band and $q$ designating the wavevector in the first Brillouin zone (BZ), $q\in [-1,1]$.
In this paper we will also use the notations ${\mathcal E}_{\nu, \pm }^{(j)}$ for the lower ("$-$") and the upper ("$+$") edges of the $\nu$-th band, and $\Delta_\nu^{(j)}={\mathcal E}_{\nu+1,-}^{(j)}-{\mathcal E}_{\nu, +}^{(j)}$ for a $\nu$-th finite gap ($\Delta_0^{(j)}$ designating the semi-infinite gaps).

For a general form of $V(x)$, the introduced linear eigenvalue
problems (\ref{eq:1deigen}) represent Hill equations~\cite{Hill}
and, as it becomes clear below, play crucial role in the systematic
construction of the nonlinear localized modes (for a review see
e.g.~\cite{BK}).

Let us first investigate the existence of the different types of
solutions of the conservative system (\ref{eq_12}),
which are homogeneous in $z$-direction. Hence, we seek for non-uniform
($x$-dependent) stationary solutions of the system
(\ref{eq_12}) in the form:
\begin{eqnarray}
\label{eq3}
    \alpha_0=a_0(x)e^{-i(k_{z} z+k_{y} y)},\quad
    \alpha_1=a_1(x)e^{-i(k_{z} z+k_{y} y)/2},
\end{eqnarray}
where we have taken into account the matching conditions
$k_{z}=k_{z0}=2k_{z1}$, $k_y=k_{y0}=2k_{y1}$ and $\Delta k=0$ .
Substituting the ansatz (\ref{eq3}) in Eqs. (\ref{eq_12}) we obtain
\begin{subequations}
\label{eq_56}
\begin{eqnarray}
\label{eq5}
    &&\Omega a_0=-\frac{d^2a_0}{dx^2}   +V(x) a_0 - 2a_1^2,
    \\\label{eq6}
    &&\frac \Omega 4  a_1=-\frac{d^2a_1}{dx^2}  +\frac14 V(x)a_1 -2a_0\bar{a}_1,
\end{eqnarray}
\end{subequations}
where $\Omega=2k_{z}- k_{y}^2$ is a spectral parameter of the problem.

System (\ref{eq_56}) is considered subject to the boundary
conditions
\begin{eqnarray}
    \label{boundary}
    \lim_{x\to\pm\infty}a_0(x)=\lim_{x\to\pm\infty}a_1(x)=0
\end{eqnarray}
corresponding to localized patterns in the $x$-direction.

Now we establish several general properties of the solutions. First, multiplying (\ref{eq5}) and (\ref{eq6}) respectively by  $d\bar{a}_0/dx$ and $d\bar{a}_1/dx$, adding conjugated equations,  and integrating over $x$ we obtain a necessary condition for existence of the localized modes:
\begin{eqnarray}
    \label{cond1}
    \int_{-\infty}^{\infty}\frac{dV(x)}{dx} \left(2|a_0(x)|^2 +|a_1(x)|^2\right)dx=0.
\end{eqnarray}

Next we notice that the Hamiltonian (\ref{H}), corresponding to the effectively one-dimensional case we are interested in, can be written down in the form $H_{1D}=H_0+H_V(0)$ where
\begin{eqnarray}
    H_0=\int_{-\infty}^{\infty}\left(-\frac 12|a_{0,x}|^2-|a_{1,x}|^2+\bar{a}_0a_1^2+ a_0\bar{a}_1^2 \right)dx
\end{eqnarray}
is the Hamiltonian of the two-wave interactions in the homogeneous medium and
\begin{eqnarray}
    H_V(\zeta)=-\frac 12 \int V(x-\zeta)\left(|a_0|^2+\frac 12 |a_1|^2 \right)dx
\end{eqnarray}
describes the effect of the periodicity.

Assuming that a solution $(a_0(x), a_1(x))$ of Eq. (\ref{eq_56})  is given, we consider
its infinitesimal shift $\zeta$ in the space, i.e.  $(a_0(x-\zeta),
a_1(x-\zeta))$. For the solution to be stable, the introduced shift
must lead to increase of the energy. Thus for such a solution we
obtain the conditions~\cite{CBKAS}
\begin{eqnarray}
\label{cond2}
    \frac{dH_V(\zeta)}{d\zeta}\Bigg|_{\zeta=0}=0\quad\mbox{and}\quad \frac{d^2H_V(\zeta)}{d\zeta^2}\Bigg|_{\zeta=0}>0\,.
\end{eqnarray}
The first one is nothing but  Eq. (\ref{cond1}) obtained above
from the energetic arguments, while the second constrain acquires
specific meaning for the definite symmetry solutions and will be
considered below.

\subsection{On phases of the stationary localized solutions}

In this section we prove that spatially localized solutions of Eqs.
(\ref{eq_56}) are real. To this end, by analogy with \cite{AKS,BK},
we set $a_j(x)=\rho_j(x)e^{i\theta_j(x)}$, define $\theta\equiv
2\theta_1-\theta_0$, and rewrite (\ref{eq_56}) in the form
\begin{subequations}
\label{explicit}
\begin{eqnarray}
\label{explicit1}
&&\rho_{0,xx}-\theta_{0,x}^2\rho_0-V_\Omega\rho_0+2\rho_1^2\cos\theta=0,
\\
\label{explicit2}
&&2\theta_{0,x}\rho_{0,x}+\theta_{0,xx}\rho_0 +2\rho_1^2\sin\theta=0,
\\
\label{explicit3}
&&\rho_{1,xx}-\theta_{1,x}^2\rho_1-\frac 14 V_\Omega\rho_1+2\rho_1\rho_0\cos\theta=0,
\\
\label{explicit4}
&&2\theta_{1,x}\rho_{1,x}+\theta_{1,xx}\rho_1-2\rho_1\rho_0\sin\theta=0,
\end{eqnarray}
\end{subequations}
where $V_\Omega\equiv V(x)-\Omega$. Multiplying Eq.
(\ref{explicit4}) by $\rho_1$ and integrating  we obtain
\begin{eqnarray}
\label{explicit5}
    \frac{d\theta_1(x)}{dx}=\frac{2}{\rho_1^2(x)}\int_C^{x}\rho_1^2(x')\rho_0(x')\sin\theta(x')dx'
\end{eqnarray}
where $C$ is a constant. Substituting this last formula in Eq.
(\ref{explicit3}) and considering the limit $x\to\pm\infty$
together with the  boundary conditions (\ref{boundary}) we find out
that a necessary condition for the existence of nonsingular
solutions is that
\begin{eqnarray}
    \lim_{x\to\pm\infty}\int_C^{x}\rho_1^2(x')\rho_0(x')\sin\theta(x')dx'=0.
\end{eqnarray}
Using the definition of $\rho_j$ and $\theta_j$ this constrain can be rewritten as follows
\begin{eqnarray*}
    &&\lim_{x\to\pm\infty}\mbox{Im}\int_C^{x}\bar{a}_0(x')a_1^2(x')dx'
    \nonumber \\
    &&= \lim_{x\to\pm\infty}\mbox{Im}\int_C^{x}\left(
    -\bar{a}_0(x')\frac{d^2a_0(x')}{dx'^2} +V_\Omega(x') |a_0(x')|^2 \right) dx'
    \nonumber \\
    &&=\lim_{x\to\pm\infty}\mbox{Re} \left[
    \bar{a}_0(C)a_{0,x}(C) - \bar{a}_0(x)a_{0,x}(x)\right] \nonumber \\
    &&= |\rho_0(C)|^2\theta_{0,x}(C)    =0
\end{eqnarray*}
where we have taken into account Eq. (\ref{eq1}) and performed
integration by parts. From the last line of the above equalities we
conclude that  $\theta_0$ does not depend on $x$, and thus without
restriction of generality it can be chosen zero so that $a_0(x)$ can
be chosen real. Then Eq. (\ref{eq5}) implies that $a_1^2$ is real.
Thus one has two options: either $a_1$ is real or $a_1$ is pure
imaginary. However, one readily concludes from the system
(\ref{eq_56}) that the second option can be reduced to the first one
by changing $a_0\mapsto -a_0$. Therefore, in what follows, we will
restrict our analysis to real stationary fields $a_{0,1}(x)$.

\subsection{Symmetry of solutions in even potentials}

For the sake of concreteness and to simplify the consideration below we set $\tilde{\varepsilon}(x)$ to have a definite parity. Then recalling the second of the constrains (\ref{cond2}) we conclude that a solution which also has a definite parity, i.e. such that $a_j^2(x)=a_j^2(-x)$, and  strongly localized about $x=0$, cannot satisfy (\ref{cond2}) if $V(x)=-V(-x)$. Moreover, for strongly localized stable solutions Eq.(\ref{cond2}) is satisfied only if $x=0$ is a minimum of the periodic function $V(x)$. Thus we can conjecture that stable solutions will be obtained centered about the minimum of the potential.

Thus, we restrict further consideration to the case
\begin{eqnarray}
\label{pot}
    V(x)=V(-x),\qquad \left.\frac{d^2V(x)}{dx^2}\right|_{x=0}>0
\end{eqnarray}
and consider $a_{0,1}$ either even or odd.

As the next step we prove that among such solutions only even are allowed, or in other words we prove the following
\begin{prop}
 If $\{a_0(x),a_1(x)\}$ is a real solution of (\ref{eq_56}) with the potential (\ref{pot}), such that $a_{0,1}^2(x)=a_{0,1}^2(-x)$, then necessarily:  $a_0(x)=a_0(-x)$ and $a_1(x)=a_1(-x)$.
\end{prop}

\begin{proof} First we prove that $a_0(x)$ is even. Assuming the opposite, i.e. that $a_0(x)$ is an odd function and integrating (\ref{eq5}) with respect to $x$ over the whole axis we obtains $P_1=0$.  Thus $a_1\equiv 0$ and (\ref{eq5}) becomes linear and hence does not allow for the existence of spatially localized solutions. The contradiction we have arrived to, proves the claim that $a_0(x)$ must be even.

Now we can prove that $a_1(x)$ is an even function. Again assuming that it is odd, what implies that $a_1(0)=0$, and designating $\displaystyle{\beta_n=\frac{1}{(2n+1)!}\left.\frac{d^{2n+1}a_1}{dx^{2n+1}}\right|_{x=0}\neq 0}$ the lowest nonzero derivative of $a_1(x)$ at $x=0$, such that $a_1(x)=\beta_nx^{2n+1}+o(x^{2n+1})$ as $x\to 0$, from (\ref{eq5}) we obtain that such a solution exists only if $V(0)-\Omega={\cal O}(x^2)$, and consequently from (\ref{eq5})  $a_0(x)={\cal O}(x^{4(n+1)})$. Now one can estimate the orders of all terms in (\ref{eq6}) at $x\to 0$ as ${\cal O}(x^{2n-1})$,  ${\cal O}(x^{2n+3})$, and ${\cal O}(x^{4n+1})$, i.e. (\ref{eq6}) cannot be satisfied in the vicinity of $x=0$. We thus again arrive at the contradiction with the supposition about the existence of the solution.
\end{proof}

\subsection{On construction of spatially localized modes}

Let us now turn to the asymptotics of the solutions at $|x|\to \infty$. Since in this limit, $q_j\to 0$, it follows directly from (\ref{eq6}) that $a_1(x)\to A_1(x)$ where $A_1(x)$ is a decaying to zero solution of the Hill equation~\cite{Hill}
\begin{eqnarray}
    \label{eq8}
    -\frac{d^2A_1}{dx^2}   + \frac 14 \left( V(x)-\Omega \right)A_1 =0.
\end{eqnarray}
This equation has decaying (growing solutions) only if  $\Omega/4$ belongs to a forbidden gap of the "potential" $\frac 14 V(x)$ (here we will not consider the cases where $\Omega$ coincides with an edge of a gap). This last requirement we formulate as $\Omega/4\in \Delta_{\nu'}^{(1)}$ for some positive integer $\nu'$. Then it follows form the Floquet theorem that
\begin{eqnarray}
\label{A1_sol}
A_1(x)= C_1e^{-\mu_{1}x} \phi_{1}(x)
\end{eqnarray}
where $C_1$ is a constant, $\mu_{1}$ is the corresponding Floquet exponent, defined by the frequency detuning to the gap, and $\phi_{1}(x)$ is the real $2\pi$--periodic function.

Let us now turn to (\ref{eq5}), and notice that it can be considered as a linear equation for the unknown $a_0(x)$. Taking into account, that for decaying solutions $a(x)$, $\Omega$ must belong to a gap $\Delta_{\nu''}^{(0)}$, of the potential $V(x)$, i.e. $\Omega\in\Delta_{\nu''}^{(0)}$ [we emphasize that the frequency $\Omega$ here is the same as in Eqs. (\ref{eq6}) and (\ref{eq8})], we arrive at a necessary condition for existence of the localized modes: there must exist a nonzero intersection between gaps  $\Delta_{\nu''}^{(0)}$ and $4\Delta_{\nu''}^{(1)}$ for at least some integers $\nu'$ and $\nu''$. Then, designating this intersection by $\Delta_g$, i.e. $\Delta_g=\Delta_{\nu''}^{(0)}\cap 4\Delta_{\nu'}^{(1)}$ and referring to it as a {\em total gap}, we conclude that if for a given $\Omega$ a localized mode exists, then $\Omega\in\Delta_g$.

Now one can write down:
\begin{eqnarray}
    a_0(x)&=&\frac{1}{W}\left(A_+(x)\int_x^{\infty} A_-(x')a_1^2(x')dx'
    + A_-(x)\int_{-\infty}^x A_+(x')a_1^2(x')dx' \right).
    \label{a_0}
\end{eqnarray}
where $A_\pm$ are the solutions of the eigenvalue problem
\begin{eqnarray}
\label{eq7}
    -\frac{d^2A_\pm}{dx^2}   + \left( V(x)-\Omega \right)A_\pm=0\,,
\end{eqnarray}
and $W=A_+(A_-)_x-(A_+)_x A_-$ is their Wronskian. Taking into account, that $\Omega$ is in a gap of the spectrum of Eq. (\ref{eq7}), we have that
$A_\pm(x)= e^{\pm \mu_{0}x} \phi_{\pm}(x)$ with $\phi_{\pm}(x)$ being real $2\pi$--periodic functions, and $\mu_0$ being the respective Floquet exponent. Thus, in the case at hand $W=-2\mu_0 \phi_{+}\phi_{-}+\phi_{+}(\phi_{-})_x-(\phi_{+})_x\phi_{-}$.

Using the fact that $\int_{-\infty}^\infty A_+(x)a_1^2(x)dx=0$ (this can be proved by substituting $a_1^2(x)$ from Eq.(\ref{eq5}) and integrating by part) the second term in (\ref{a_0}) can be rewritten in the same integration interval as the fist one
\begin{eqnarray}
    a_0(x)&=&\frac{1}{W}\left(A_+(x)\int_x^{\infty} A_-(x')a_1^2(x')dx'
    - A_-(x)\int_{x}^{\infty} A_+(x')a_1^2(x')dx' \right).
    \label{a_0_new}
\end{eqnarray}

\subsection{Numerical study of localized modes}

Now one can use the shooting  method to construct localized solutions of the system (\ref{eq_56}). As the first step one has to ensure that for a chosen structure the total gap $\Delta_g$ exists. As the next step for a chosen $\Omega\in \Delta_g$, the Floquet exponents $\mu_{0,1}$ and corresponding Bloch functions $\phi_1$ and $\phi_{\pm}$ must be computed.
Next, starting from some point $x=x_0$ far enough from the origin, where equation (\ref{eq6}) is effectively linear, one can approximate the function $a_1(x)$ on the interval $x_0<x<\infty$ by its linear asymptotics $A_1(x)$, which is determined by (\ref{A1_sol}) with some small initial amplitude $C_1$. Substituting $A_1(x)$ into equation (\ref{a_0_new}) and computing the integrals one has to obtain the value of the function $a_0(x_0)$ which now can be used to find $a_1(x)$ at the subsequent step of the spatial grid, $x=x_0-dx$ by solving Eq.(\ref{eq6}) (we use the Runge-Kutta method). Finally, by varying the shooting parameter $C_1$ one has to satisfy the condition $da_{0,1}/dx=0$ at the origin $x=0$ (recall that we are looking for even solutions).

Let us now show several possible solutions of the system (\ref{eq_56}).
For the numerical simulations we assume the specific form of the dielectric permittivity: $V(x)=-V_0\cos(2x)$.
In the left panel of  Fig.\ref{fig_1} an example of the band-gap structure of Eqs. (\ref{eq_56}) with a total gap (it is indicated by the shadowed region) is shown. In this particular case this is an overlap of the first gaps of each of the equation.
The profiles of the corresponding solutions for different $\Omega$ inside the total gap are shown in the right panels of Fig.\ref{fig_1}.

\begin{figure}
\vspace{0.5cm}
\epsfig{file=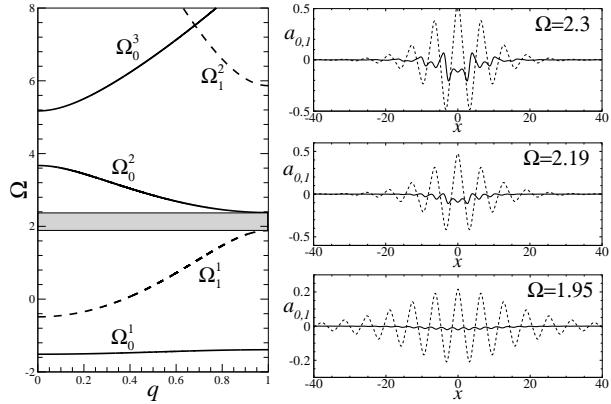,width=8cm}
\vspace{0.5cm}
\caption{In the left panel the band-gap spectra $\Omega^{(\nu)}_0(q)$ for the pump  (solid line) and $4\Omega^{(\nu)}_1(q)$ for the signal (dashed line) are shown.
$\nu$ determines the corresponding band number.
The shadowed domain shows the interval $\Delta_{g}=[1.89; 2.36]$ where the first gaps for the pump and the signal have overlapping and, consequently, where localized solutions of (\ref{eq_56}) can be found.
In the right panels the profiles of the localized solutions $a_{0}(x)$ (solid lines) and $a_1(x)$ (dashed lines) with $\Omega\in\Delta_g$. Here $V_0=4$.}
\label{fig_1}
\end{figure}

In all the pictures presented one observes that the amplitude of the signal is appreciably larger than the amplitude of the pump. Moreover, the amplitude of the mode itself (it is two-component in our case) decays and the pump amplitude become very small as the frequency approaches the bottom of the total gap. Notice that the top and bottom of the total gap coincide with the top of the gap for the pump signal and with the bottom of the gap for the signal. We also observe that the pattern of the pump is more sophisticated than the pattern of the signal component.

To investigate the stability of obtained solutions we integrate numerically the evolution equations  (\ref{eq_12}) taking $\Delta k=0$ and starting with initial profiles presented in the right panels of Fig.~\ref{fig_1}.
In Fig.~\ref{dyn} we show the propagation of the localized modes for three different $\Omega$. We observe that the modes excited close to the bottom of the total gap are  stable while   the solutions in the vicinity of the top of the gap display unstable behavior.

\begin{widetext}
\vspace{3cm}
\begin{figure}
\epsfig{file=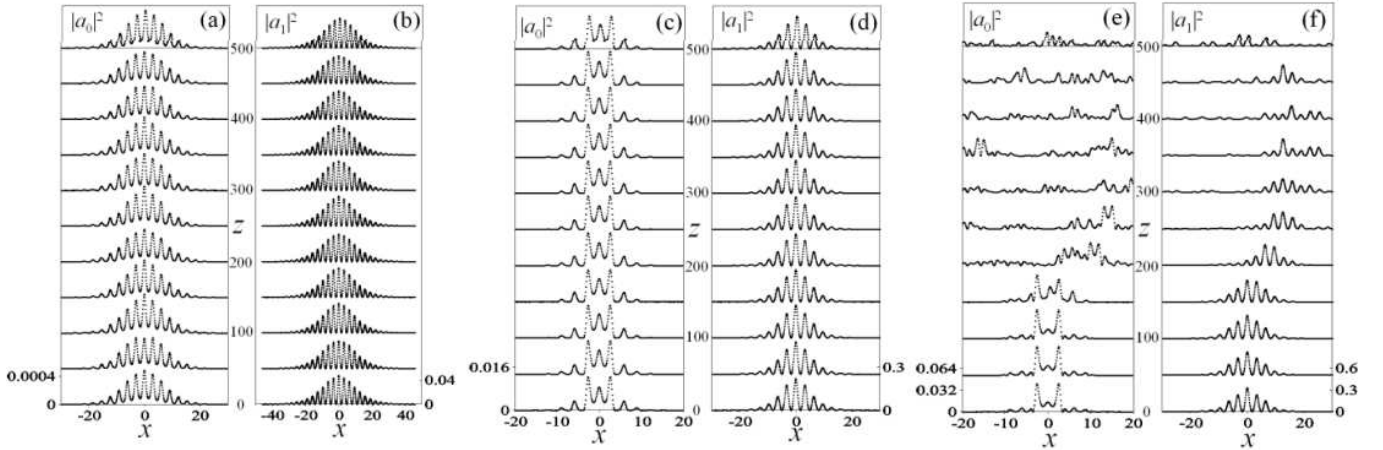,width=18cm}
\caption{Dynamics of the modes, whose initial profiles are taken from the right panel of Fig.~\ref{fig_1} and perturbed by 1\% of their amplitudes.
In [(a),(b)], [(c),(d)] and [(e),(f)] parameter  $\Omega=1.95; 2.19; 2.3$, respectively.}
\label{dyn}
\end{figure}
\end{widetext}

\section{Coherent structures in the optical parametric oscillator}

\subsection{The model}

Now we proceed with the study of possible coherent patterns in an optical parametric oscillator and as the first step we recall derivation of the model equations, following~\cite{Tlidi}.
We consider the simplified case where in transverse direction there exists dependence only on one variable [we call it $x$: $\tilde{\varepsilon}\equiv\tilde{\varepsilon}(x)$] and use the abbreviated notation $\alpha_{0,1}(z,\tau)=\alpha_{0,1}(z,x,\tau)$.
To this end we impose the boundary conditions on (\ref{eq_12}), (\ref{eps01}) assuming that the cavity is located on the interval $0<z<\ell$:
\begin{subequations}
\label{boundary12}
\begin{eqnarray}
\label{boundary1}
    \alpha_0\left(0,\tau\right)&=&\alpha_0^{in}+R_0e^{i\theta_0} \alpha_0(\ell,\tau-T),
    \\
    \label{boundary2}
    \alpha_1(0,\tau)&=&R_1e^{i\theta_1} \alpha_1(\ell,\tau-T)
\end{eqnarray}
\end{subequations}
and allowing the amplitude of the driven field to depend on the transverse coordinates $x$, i.e. $\alpha^{in}\equiv \alpha^{in}(x)$.
Here $T=\ell/c$ is the delay time with $c$ being the velocity of light, $\theta_{0,1}$ are the detuning parameters and $R_{0,1}$ are the reflectivity factors (in the following we will consider only the case of exact phase matching, $\Delta k=0$).
In the small amplitude limit, $|\alpha_{0,1}|\ll 1$, one can write down approximated solutions for pump and signal for the reduced Maxwell equations (\ref{eq_12}) as  the first order terms of the Mac-Laurin expansion
\begin{subequations}
\begin{eqnarray}
&&\alpha_0(\ell,\tau)= \alpha_0(0,\tau)+\left[\frac i2\partial_x^2\alpha_0(0,\tau)
+ i\tilde{\varepsilon}\alpha_0(0,\tau)+i\alpha_1^2(0,\tau)\right]\ell\,,
    \label{MacLau_1}\\
&&\alpha_1(\ell,\tau)=\alpha_1(0,\tau )+ \left[i\partial_x^2\alpha_1(0,\tau)
+\frac i2\tilde{\varepsilon}\alpha_1(0,\tau) +2i\alpha_0(0,\tau) \bar{\alpha}_1(0,\tau)\right]\ell\,.
\label{MacLau_2}
\end{eqnarray}
\end{subequations}
Assuming that the time variation of the solutions is very slow as compared to the delay
\begin{eqnarray}
T\p_\tau \alpha_{0,1}(0,\tau)= \alpha_{0,1}(0,\tau+T)-\alpha_{0,1}(0,\tau)
\label{d_tau}
\end{eqnarray}
and considering limits of large reflectivity $R_{0,1}\approx 1$ and small detuning $\theta_{0,1}\ll 1$
one comes to the model
\begin{eqnarray}
    &&\partial_\tau A_0=E(x)-(\gamma+i\Delta_0)A_0
       -\frac i2 {\cal L}_0
    A_0-A_1^2, \label{A1}
    \\
    &&\partial_\tau A_1=-(1+i\Delta_1)A_1
    -i {\cal L}_1A_1 +2A_0\bar{A}_1 \label{A2}
\end{eqnarray}
(see also Refs.~\cite{ZMT,BRCTT}, Refs.~\cite{Conti} and \cite{Egorov}, where the homogeneous counterpart of this system was derived respectively withing the framework of the multiple-scale expansion for excitations with smooth envelopes and as a continuum approximation for the cavity solitons in a periodic media, originally described by the coupled discrete equations).
To derive these equations one uses the renormalized fields
\begin{eqnarray}
\label{scale}
\frac{\ell}{1-R_1}\left(
i\alpha_0,\alpha_1, -\frac{V(x)}2 , \frac{i\alpha_0^{in}}{1-R_1}\right)\to\left(A_0,A_1,\tilde\varepsilon,E\right),
\end{eqnarray}
 rescaled time and space variables:
\begin{eqnarray}
\label{scale_tx}
\left(\frac{T}{1-R_1}\p_\tau, \frac{\ell}{1-R_1}\partial_x^2\right)\to\left(\p_\tau,\partial_x^2\right),
\end{eqnarray}
as well as the constants
$ \gamma=(1-R_0)/(1-R_1)$ and $\Delta_{0,1}=\theta_{0,1}/(R_1-1)$.
The linear operators ${\cal L}_{0,1}$ have been defined in (\ref{lin_op}).

\subsection{Nonlasing states}
\label{subsec:non-lasing_state}

These states correspond to $A_1=0$ and thus the respective $A_0$ solves the stationary ($\p_\tau=0$) equation
\begin{eqnarray}
\label{eq0}
    {\cal L}_0A_0^{(st)} + 2(\Delta_0-i\gamma)A_0^{(st)} + 2iE(x)=0.
\end{eqnarray}

Let us consider the driven field having the same periodicity as the medium, i.e. $E(x)=E(x+\pi)$.
Since the Bloch states constitute a complete orthonormal basis one can represent
\begin{eqnarray}
    E(x)=\sum_{\nu,q}e_{\nu q}\varphi_{\nu q}^{(0)}(x).
\end{eqnarray}
Looking for the respective stationary solution $A_0^{(st)}$ also in a form of the expansion over the Bloch functions one readily obtains:
\begin{eqnarray}
    A_0^{(st)}(x)=-2i\sum_{\nu,q}\frac{e_{\nu q}}{{\cal E}_{\nu q}^{(0)}+2(\Delta_0-i\gamma)}\varphi_{\nu q}^{(0)}(x).
\end{eqnarray}
The obtained field pattern has an interesting particular limit which occurs when $-2\Delta_0$ coincides with one of the band edges of the pump signal, i.e. when $-2\Delta_0={\cal E}_{\nu \pm}^{(0)}$. Then $A_0^{(st)}(x)\equiv E(x)/\gamma$, i.e. the pattern of the pump signal is the same as the pattern of the driven field (with accuracy of a factor $1/\gamma$).

\subsection{Lasing states, particular case 1}
\label{subsec:lasing_state1}

In the presence of a periodic potential of a general kind, finding  explicit solutions for lasing states even in the simplest statement where $A_1=$const seems to be an unsolvable problem. A progress however is possible for periodic functions $V(x)$ of a specific type. In order to show this we employ a kind of "inverse engineering"~\cite{BK} and find an explicit form of $V(x)$ assuming that $A_1$ is a given complex constant: $A_1=\rho e^{i\varphi}$  (here both $\rho$ and $\varphi$ are real constants). Then looking for  the pump filed in the form $A_0(x)=a_0(x)e^{2i\varphi}$, where $a_0(x)$ is a real function,  we obtain from (\ref{A2})
\begin{eqnarray}
\label{a00}
    a_0(x)=\frac12\left[1+i\Delta_1+\frac i4 V(x)\right]
\end{eqnarray}
which after substituting in (\ref{A1}) yields the system of two real equations:
\begin{eqnarray}
\label{V1}
    \Delta_0&+&\gamma\Delta_1+\frac{\gamma+2}{4}V(x)-2E_I=0,
    \\
    \label{V2}
    \frac{d^2V(x)}{dx^2}&-&2(\Delta_0+2\Delta_1)V(x)-V^2(x)\nonumber\\
    &+&8(\gamma-\Delta_0\Delta_1+2\rho^2-2E_R)=0
\end{eqnarray}
where $E_I=$Im$\left(Ee^{-2i\varphi}\right)$ and $E_R=$Re$\left(Ee^{-2i\varphi}\right)$. Eq. (\ref{V2}),  viewed as an equation with respect to $V(x)$, can be solved in quadratures:
\begin{eqnarray}
\label{x_V}
 x&=&\int_{0}^{V} \left[\frac 23 V^3+2(\Delta_0+2\Delta_1)V^2
 - 16(\gamma-\Delta_0\Delta_1+2\rho^2-2E_R)V\right]^{-1/2}dV
\end{eqnarray}
where without loss of generality we have assumed that $V(x)>0$ with $V(0)=0$ and $V_x(0)=0$, i.e. that $x=0$ is a local minimum of the potential. This last condition implies existence of the threshold for $E_R$, which now must satisfy the inequality $E_R>(\gamma-\Delta_0\Delta_1+2\rho^2)/2$.

The obtained expression   (\ref{x_V}) depends on many physical parameters. Therefore, to simplify the analysis, we narrow the class of the potentials considering a case example of a one-parametric family of the periodic functions $V(x)$:
\begin{eqnarray}
\label{elliptic}
        V(x)= V_0(k)\mbox{sn}^2\left( \beta(k) x, k\right),
\end{eqnarray}
where sn$(x,k)$ is a Jacobi elliptic function with the modulus $k$,   $\beta(k)=\frac 2\pi K(k)$,   $V_0(k)=6k^2\beta^2(k)$ is the potential amplitude parametrized by $k$, and $K(k)$ is a complete elliptic integral of the first kind.

We observe that in practical terms the profile described by  (\ref{elliptic}) is not too sophisticate: it can be reproduced with very high accuracy by only  two or three harmonics for the elliptic parameter $k$ not too close to one (see the discussion in \cite{BK}).

Assuming that the refractive index of the cavity has the profile (\ref{elliptic}), one finds that in order to support the stationary modes one has to apply the following form of the driven filed
\begin{eqnarray}
    \label{E}
        E(x)=e^{2i\varphi}\left[E_R+iE_{I0} +iE_{I1}\mbox{sn}^2(\beta x, k) \right]
\end{eqnarray}
where $E_{I0}=(\Delta_0+\gamma\Delta_1)/2$ and $E_{I1}=-(\gamma+2)V_0(k)/8$ are fixed by the periodic structure and the cavity properties, while $E_R$ is the control parameter which can be changed.
Then the field pattern is obtained from (\ref{a00}):
\begin{eqnarray}
\label{rho}
    \rho^2=-\frac34k^2\beta^4(k)-\frac12(\gamma-\Delta_0\Delta_1)+E_R,
    \\
    \label{a0}
    a_0(x)=\frac12\left[1+i\Delta_1+\frac i4 V_0\mbox{sn}^2(\beta (k)x, k)\right].
\end{eqnarray}
provided the elliptic modulus $k$ is expressed in terms of the detunings $\Delta_{0,1}$ through the implicit equation
\begin{eqnarray}
\label{k}
(1 + k^2) \beta^2(k)= -\Delta_1 - \Delta_0/2.
\end{eqnarray}
 The last equation means that $V_0$ and $\Delta_{0,1}$ are not independent parameters and have to be matched and that the solutions we are dealing with exists only if at least one detuning is negative, or more precisely if $ 2\Delta_1 + \Delta_0<0$.

 Now, taking into account that $\rho$ is a real constant one comes to the following threshold value $E_{thr}$ for the control parameter $E_R$
\begin{eqnarray}
\label{E_0}
    E_R>E_{thr}=\frac34k^2\beta^4(k)+\frac12(\gamma-\Delta_0\Delta_1).
\end{eqnarray}
As it is clear $E_{thr}$ corresponds to $\rho=0$.
In Fig.\ref{fig_E0_k} the dependence of the critical value of the driven field $E_{thr}$ on
the depth of modulation  of the refractive index, $V_0$,
is shown.
\begin{figure}
\epsfig{file=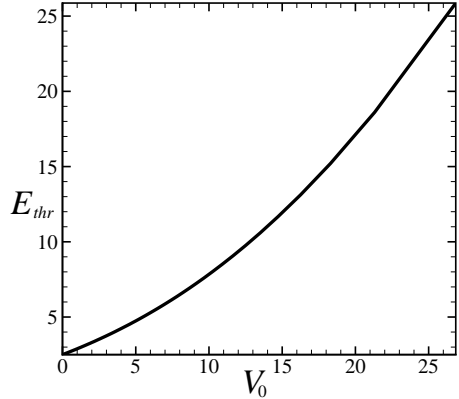,width=6cm}
\caption{The dependence $E_{thr}$ {\it vs} $V_0(k)$ for  $\gamma=1$ and $\Delta_1=1$.}
\vspace{1.5cm}
\label{fig_E0_k}
\end{figure}

In Fig.\ref{fig_2} we show several patterns for different values of the amplitude of modulation of the refractive index, $V_0$, and of the control parameter, $E_R$. Using initial dynamical equations (\ref{A1}), (\ref{A2}) we have checked that all the solutions presented are dynamically stable against initial small (of order of 5\% of the intensity) perturbations.

\begin{figure}
\epsfig{file=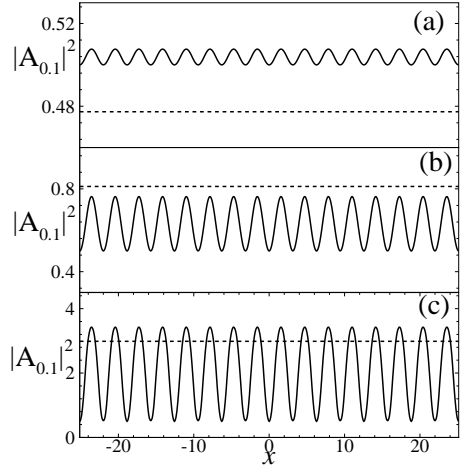,width=6cm}
\caption{Patterns of the pump $|A_0|^2$ (solid line) and signal $|A_1|^2$ (dashed line) waves determined by (\ref{rho}), (\ref{a0}) for different $V_0$ and $E_R$. In (a) $V_0=0.06$, $E_R=3$, $\Delta_0=-4.03$, in (b) $V_0=1.73$, $E_R=4$, $\Delta_0=-4.88$. in (c) $V_0=10.24$, $E_R=11$, $\Delta_0=-9.63$.
The other parameters are $\gamma=1$ and $\Delta_1=1$.}
\label{fig_2}
\end{figure}

\subsection{Lasing states, particular case 2}
\label{subsec:lasing_state2}

The approach, based on the inverse engineering and described in Sec.~\ref{subsec:lasing_state1}, can be generalized. Indeed, let us look for a plane wave solution of the form
\begin{eqnarray}
    A_0=\left[\frac12+i \tilde{A}_0(x)\right]e^{-2i\omega\tau},\quad A_1=\tilde{A}_1(x)e^{-i\omega\tau}
\end{eqnarray}
where $\omega$ is a frequency to be determined and $\tilde{A}_{0,1}(x)$ are real functions depending only on $x$ (c.f. (\ref{a00}) where $\tA_0=\Delta_1/2+V(x)/4$). Here we use the form of the potential as in the Sec.II~F, namely $V(x)=-V_0\cos(2x)$.

Let us now require $\tA_0$ to be a Bloch state of the linear eigenvalue problem
\begin{eqnarray}
\label{eq1:a0}
    \frac 12 {\cal L}_0\tA_0+\Delta_0\tA_0=2\omega \tA_0\,.
\end{eqnarray}
As it is clear $2\omega$ must be a frequency bordering a gap, i.e. there must be $\omega=\frac 14 {\cal E}_{\nu \pm}^{(0)}+\frac 12 \Delta_0$,  since otherwise the eigenvalue $\tA_0$ is complex. The amplitude of $\tA_0(x)$ is a free parameter, so far, because (\ref{eq1:a0}) is linear. Let us now fix it by requiring  the frequency $\omega$ to border a gap of another linear eigenvalue problem, which reads
\begin{eqnarray}
\label{eq1:a1}
    \left[ {\cal L}_1-2\tA_0(x)\right]\tA_1+\Delta_1\tA_1=\omega \tA_1\,.
\end{eqnarray}
Then  $\tA_1$ is a Bloch state of this new  problem where the effective periodic potential is given by  $\tilde V(x)=V(x)/4-2\tA_0(x)$. Designating the gap edges of the Hill operator (\ref{eq1:a1}) by $\tilde{\cal E}_{\nu \pm}^{(1)}$ (notice that they depend on the amplitude of $\tA_0(x)$), we deduce that the amplitude of the pump signal $\tA_0(x)$ is determined from the relation
\begin{eqnarray}
\label{cond}
\tilde{\cal E}_{\nu\pm}^{(1)}+\Delta_1=\frac 14 {\cal E}_{\nu \pm}^{(0)}+\frac 12 \Delta_0\,,
\end{eqnarray}
which can be solved numerically.
The system (\ref{A1}), (\ref{A2}) is satisfied if the complex field is chosen in the form
\begin{eqnarray}
    E(x)&=&\left[\frac \gamma 2 +\tA_1^2(x)
    + i\left(\frac 14 V(x) - \omega +\frac 12\Delta_0+ \gamma \tA_0(x)\right)\right]e^{-2i\omega\tau}.
    \label{E_x}
\end{eqnarray}

In Fig.\ref{fig_las_stII} we illustrate a band structure of (\ref{eq1:a1}) (panel a) as well as effective potential $\tilde V(x)$ (panel b) with $\tA_0(x)$ taken at the band edge ${\cal E}_{0+}^{(0)}$ and having the amplitude $1$.
The patterns of $\tA_1(x)$ corresponding to the three lowest band edges denoted as A, B and C and the respective profiles of driving field $|E(x)|^2$ are shown in Fig.~\ref{fig_las_stII} (c)-(h).

\begin{figure}
\epsfig{file=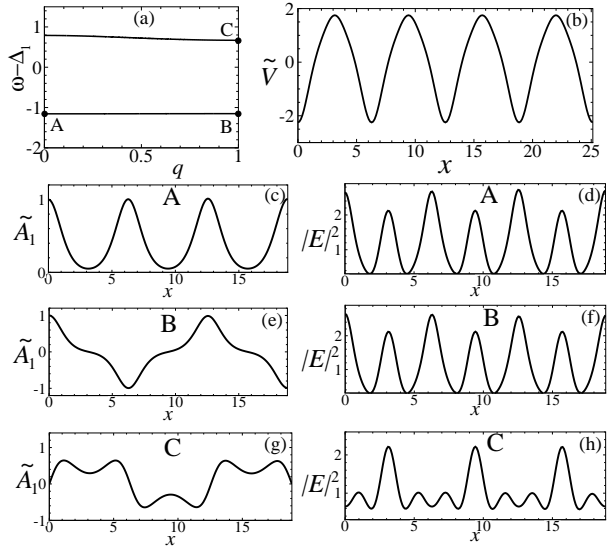,width=8cm}
\caption{In (a) the band structure of the linear eigenvalue problem (\ref{eq1:a1}) and
in (b) the profile of the corresponding effective potential, $\tilde V(x)$, are shown.
In [(c),(d)], [(e),(f)], and [(g),(h)] the solutions, $\tA_1(x)$, corresponding to the first three edges of the bands denoted by points A (the edge ${\cal E}^{(1)}_{0-}$), B (the edge ${\cal E}^{(1)}_{0+}$), and C (the edge ${\cal E}^{(1)}_{1-}$) as well as corresponding profiles of the intensity of the driving field $|E(x)|^2$ calculated from (\ref{E_x}) are shown. The profile of $\tA_0(x)$ is calculated at the upper edge of the first band ${\cal E}^{(0)}_{0+}$ of the linear eigenvalue problem (\ref{eq1:a0}).
The other parameters are $\gamma=1$ and $V_0=1$.}
\label{fig_las_stII}
\end{figure}

\section{Concluding remarks}

In this paper we have carried out analytical and numerical studies of
localized modes in optical parametric amplifiers and of periodic patterns in parametric
oscillators generated by a properly chosen driving field in presence of diffraction and transverse periodic modulations of the refractive index.

Using Bloch waves approach we have found that the presence of refractive
index modulations drastically affect the existence and properties of
localized structures in the both systems. In particular, spatial
modulations lead to the formation of transverse phase locked (real
fields) signal and pump in the form of a localized modes (gap solitons).
Conditions required on the amplitude of modulations and external pump field for
the existence of localized patterns in the parametric oscillators have been deduced.
We have checked the stability of the obtained solution
by means of integration of respective dynamical equations.
It as to be emphasized however that the thorough analysis of the stability of gap solitons, which is complicated by the necessity of numerical generation of exact gap solitons, is left for further studies. In the present work we also did not discuss gap solitons of high order gaps, whose existence is strongly constrained (and may be even inhibited) by the necessity of the overlapping of  the spectra of two components, i.e. of the first and second harmonics (notice that the respective constrains do not exist in the theory of the single-component solitons in Kerr media). And meantime, the method elaborated for the particular case, where signal and idle waves were equal, allows straightforward generalization to the situation where resonant interaction among three different waves occurs.

\acknowledgments

VAB was supported by the FCT grant SFRH/BPD/5632/2001. VVK
acknowledges support from Ministerio de Educaci\'on y Ciencia (MEC,
Spain) under the grant SAB2005-0195. The work of VAB and VVK was
supported by the FCT and European program FEDER under the grant
POCI/FIS/56237/2004. Cooperative work was supported by the bilateral
program Ac\c{c}\~ao Integrada Luso-Francesa. The IRCICA and CERLA
are supported in part by the ``Conseil R\'{e}gional Nord Pas de
Calais'' and the ''Fonds Europ\'{e}en de D\'{e}veloppement
Economique des R\'{e}gions''.

\end{document}